\documentclass{llncs}

\input glyphtounicode
\usepackage[T1]{fontenc}
\usepackage{graphicx}
\usepackage{comment,array}
\usepackage{cite}
\usepackage{algorithmic}
\usepackage[linesnumbered,ruled,vlined]{algorithm2e}
\usepackage[latin1]{inputenc}
\usepackage[tight]{subfigure}
\usepackage{amssymb}
\usepackage{amsmath}
\usepackage{bbm}
\usepackage{microtype}
\usepackage[
    final,
    bookmarks=true,
    bookmarksnumbered=true,
    linktocpage,
    colorlinks=true,
    urlcolor=black,
    citecolor=black,
    linkcolor=black,
    filecolor=black,
    linktoc=all
]{hyperref}
\usepackage[all]{hypcap}

\usepackage{xspace}
\usepackage{acronym}

\acrodef{FIFO}{First-In, First-Out}
\acrodef{PAS}{Patience-Aware Scheduling}
\acrodef{EAS}{Expectation-Aware Scheduling}

\newcommand{\FIFO}{\ac{FIFO}\xspace}
\newcommand{\PaS}{\ac{PAS}\xspace}
\newcommand{\EaS}{\ac{EAS}\xspace}

\newcommand\blfootnote[1]{%
  \begingroup
  \renewcommand\thefootnote{}\footnote{#1}%
  \addtocounter{footnote}{-1}%
  \endgroup
}

\title{Patience-aware Scheduling for Cloud Services:\\ Freeing Users from the Chains of Boredom}

\author{}
\institute{Carlos Cardonha$^1$, Marcos D. Assun\c{c}\~{a}o$^1$,\\ Marco A. S. Netto$^1$, Renato L. F. Cunha$^1$, Carlos Queiroz$^2$\\
$^1$IBM Research Brazil\\
$^2$IBM Research Australia
}

\begin{document}

\maketitle

\blfootnote{The final publication is available at link.springer.com}

\begin{abstract}
Scheduling of service requests in Cloud computing has traditionally focused on the reduction of pre-service wait, generally termed as waiting time. Under certain conditions such as peak load, however, it is not always possible to give reasonable response times to all users. This work explores the fact that  different users may have their own levels of tolerance or patience with response delays. We introduce scheduling strategies that produce better assignment plans by prioritising requests from users who expect to receive the results earlier and by postponing servicing jobs from those who are more tolerant to response delays. Our analytical results show that the behaviour of users' patience plays a key role in the evaluation of scheduling techniques, and our computational evaluation demonstrates that, under peak load, the new algorithms typically provide better user experience than the traditional FIFO strategy.
\end{abstract}

% -------------------------------------------------------------------------------------------
\section{Introduction}
\label{sec:intro}

Job schedulers are key components of Clouds as they are responsible not only for assigning user tasks to resources but also for notifying management systems on when resources need to be allocated or released. These resource allocation decisions, specially on when to allocate additional resources, have an impact on both provider costs and user experience, and are particularly relevant to manage resources under peak loads.

Traditionally, job schedulers do not take into account how users interact with services. They optimise system metrics, such as resource utilisation and energy consumption, and user metrics such as response time. However, understanding interactions between users and a service provider over time allows for custom optimisations that bring benefits for both parties. Such interactions are becoming more pervasive due to the large number of users accessing Cloud services via mobile devices and analytics applications that require multiple service requests.

In this article we propose scheduling strategies that take into account users' expectations regarding response time and their patience when interacting with Cloud services. Such strategies are relevant mainly to handle peak load conditions without the need to allocate  additional resources for the service provider. Although elasticity is common in a Cloud setting, resources may not be available quickly enough and their allocation can incur additional costs that may be avoidable. The main contributions of this paper are:

\begin{itemize}
    \item The introduction of a \PaS strategy and an \EaS strategy for Cloud systems;
    \item Analytical comparisons between the \EaS  strategy and the traditional \FIFO scheduling strategy;
    \item Evaluation of the proposed strategies  and a detailed discussion on when they bring benefits for users and service providers.
\end{itemize}

% -------------------------------------------------------------------------------------------

\section{Proposed Scheduling Strategies}
\label{sec:scheduling}

This section describes the proposed strategies, \PaS and \EaS, and presents several analytical results that compare \EaS with \FIFO. We chose \FIFO for comparison because it is one of the most used scheduling techniques that explores fairness of users by scheduling requests as they arrive in the system.

\subsection{Common Characteristics Shared by \PaS and \EaS}

This work considers Cloud services (\emph{e.g.} data analytics, Web search, social networks) that back applications running on mobile devices and desktops, most of which are highly interactive and iterative. Users, consciously or not, interact with a service provider multiple times when using their applications. Service performance over time usually shapes the users' expectations on how it is likely to perform in the future. The service provider stores information on how its service responded to user requests and uses this information to gauge her expectations and patience. 

\PaS and~\EaS utilise user expectation to schedule service requests on the Cloud's resources. Both strategies  share the following common goals:
\begin{itemize}
    \item Minimise the number of users abandoning the service;
    \item Maximise the users'  level of happiness with the service;
    \item Perform such optimisations  without adding new resources to the service.
\end{itemize}
We remark that an incoming job request will be directly assigned if there are available resources
in the service provider. Therefore, choosing among~\FIFO, \PaS, and \EaS
becomes more crucial during peak load.

\subsection{Patience-aware Scheduling}
\label{sec:patienceaware}
\PaS{}  has the goal of serving first users whose patience levels are the lowest when interacting with the Cloud service. When new requests arrive, the algorithm sorts the tasks in its waiting queue according to the \emph{Patience} of their users (in ascending order), and when a new resource is freed, the request positioned in the head of the waiting list is assigned to it.

An adequate estimate of how the user's happiness  level and the user's tolerance curves behave is very important for the evaluation of the proposed strategies. In our implementation of~\PaS and in our computational evaluation, we employed the definition of  user's patience suggested by Brown \emph{et al.}~\cite{brown2005statistical}. Patience is hence given by the ratio of the time a user expects to wait  for results to the time the user actually waits for them:
\[
    \text{Patience} = \frac{\text{Expected Response Time}}{\text{Actual Response Time}}
\]

\subsection{Expectation-aware Scheduling}

\EaS has the goal of serving first  requests associated to users whose response time expectations are translated into ``soft'' deadlines that are positioned earlier in time. The difference between \EaS and traditional deadline-based algorithms lies in the nature of the ``buffer'' adding to the minimum response time, as it changes over time and is related to users' patience. 

\EaS sorts service requests in the waiting queue according to their users' expectation, given by: 
\[
    \text{Expectation} = \text{arrival time} + \text{expected response time},
\]
where \emph{arrival time} is the time at which the job arrived on the waiting queue and \emph{expected response time} is the time that the service provider need to complete the task. \EaS, then, schedules the job with the least \emph{expectation} when a new resource is freed.

\subsection{Analytical Investigation of the \EaS Strategy}

In this section, we analyse \EaS from a theoretical point of view by comparing it with \FIFO. The notation used throughout this section is presented in Table~\ref{tab:list-of-symbols} and explained in more details in the section below.

\subsubsection{Notation.}

\begin{table}
\begin{center}
\caption{List of symbols.}
\begin{tabular}{r|l}
    $\mathcal{T}$       & sequence of arriving tasks \\
    $t $                & task in $\mathcal{T}$ \\
    $a(t)$              & arrival time of task $t$\\
    $s(t)$              & time at which task~$t$ starts to be processed \\
    $r(t)$              & response time for~$t$\\
    $\mathcal{U}$       & set of users \\
    $u$                 & user in $\mathcal{U}$ \\
    $e(u,t)$              & time difference between user $u$'s expectation and real response time\\
    $\Delta(t)$         & processing time of task $t$\\
    $u(t)$              & user that submitted task $t$\\
    $w(u)$            & time tolerance that $u$ has for the results of some service\\
    $m$              & parallelism capacity of service provider \\
    $h(u)$            & user $u$'s level of happiness with service provider\\
    $c(u)$              & minimum level of happiness at which user $u$ still uses a service\\
    $i(u,e(u,t))$  & variation of $h(u)$ according to $e(u,t)$\\
    $j(u,t,e(u,t))$     & variation of $w(u, t)$ according to $e(u,t)$\\
    $Z$                 & closed interval $[0,1] \in \mathbb{R}$\\
    $Z^{|\mathcal{U}|}$ & user happiness state space\\
\end{tabular}
\end{center}
\label{tab:list-of-symbols}
\end{table}

Let~$\mathcal{U}$ be the set of users of a service provider.
Let~$\mathcal{T}$ denote the sequence of job requests being submitted, 
where each $t \in \mathcal{T}$ arrives at time
$a(t) \in \mathbb{R}^{+}$ and has processing time $\Delta(t)
\in \mathbb{R}$. Task~$t$ is submitted by user~$u(t)$, who is expecting to wait
an amount of time $w(u(t)) \in \mathbb{N}$ in addition to~$\Delta(t)$, \emph{i.e.}, $w(u)$ 
denotes~$u$'s tolerance with response delays. 
The service provider has a dispatching algorithm responsible for the assignment  
of each incoming task to one of its~$m$  indistinguishable and non-preemptive processors.

Let us denote by $s(t) \in \mathbb{R}^{+}$ the time at which task~$t$ starts to be processed. The
 response time for task~$t$ is
given by $r(t) = (s(t) - a(t)) + \Delta(t)$, and 
$e(u(t),t) = r(t) - (\Delta(t) + w(u(t)))$ denotes 
the amount of time by which the response time differs from~$u(t)$'s original expectation.

We denote user $u$'s level of happiness by~$h(u) \in \left[0,1\right]$, a linear scale where $h(u) = 0$ and $h(u) = 1$ indicates that~$u$ is absolutely discontent and happy, respectively. We assume that~$u$ stops sending
requests  as soon as~$h(u)$ reaches a value below some critical value~$c(u)$ in~$\left[0,1\right]$.
We say that user~$u$ is active if $h(u) > c(u)$.
The impact that~$e(u,t)$ has on~$h(u)$ is formulated  by  function $i: \mathcal{U}  \times \mathbb{R} \rightarrow \mathbb{R}$, and the impact that~$e(u,t)$ has on~$w(u)$ is described by  some function $j: \mathcal{U} \times \mathcal{T}  \times \mathbb{R} \rightarrow \mathbb{R}$.
If we assume that~$i(u,e(u,t))$ and~$j(u,t,e(u,t))$ are addictive
factors, then, after the computation of some task~$t$, the happiness level of user~$u(t)$ 
 will be given by $h(u) + i(u,e(u,t))$, while $u(t)$'s patience level becomes
$w(u(t)) + j(u(t),t,e(u(t),t))$.

\subsubsection{Optimisation Criteria.}

Let $Z$ denote the closed interval~$[0,1] \subset \mathbb{R}$.  We say that a vector $s \in Z^{|\mathcal{U}|}$ denotes a service provider's  \emph{user happiness state} if $s_x = h(u_x)$ $\forall u_x \in \mathcal{U}$, $1 \leq x \leq \mathcal{U}$.  In order to evaluate and compare different scheduling strategies, we have to define a cost function $c: Z^{|\mathcal{U}|} \rightarrow \mathbb{R}$.  It is clear that the definition of a proper cost function depends on the optimisation criteria one
wants to establish.

For our theoretical analysis, we will consider two optimisation goals.
The first one is the \emph{maximisation of the overall happiness of users}, where
service providers should try to  reach states $s \in Z^{|\mathcal{U}|}$ of maximal $L^1$-norm.
The other criteria consists of the \emph{maximisation of active users}, where service providers
try to keep as many active users as possible. Formally, a state $s \in Z^{|\mathcal{U}|}$ satisfying
this second goal  is associated to a vector $s' \in Z^{|\mathcal{U}|}$  such that $s'_x =s_x$ if $s_x \geq c(u_x)$, 
 $s'_x = 0$ otherwise, and $||s'||_{0}$ is maximal.

In the following sections, we investigate two  scenarios for the problem and discuss the situations where 
these different optimisation criteria may be employed and how the proper choice of a scheduling algorithm
depends strongly on the behaviour of functions~$i$ and~$j$.

\subsubsection{Batch Requests.}

We consider initially how scheduling strategies affect the user happiness states
when we take into account a single batch of job requests. We assume here that
each user submits a single request, and therefore we do not investigate variations
of~$w(u)$.
Single batch analysis is interesting because it is the only reasonable option
when requests do not arrive in a periodic fashion and users' profiles are unknown, \emph{i.e.}, when we do not have an exact
idea about their patiences' levels and behaviours.  The
optimisation criteria in this section will be the $L^1$-norm of the user
happiness state vector.

Let us consider the family of scenarios where each task in~$\mathcal{T}$  consumes 
time~$\Delta$, and let~$t_x, t_y \in \mathcal{T}$ be such 
that $x+m < y$ and $a(t_x) + w(u_x) > a(t_y) + w(u_y)$.

If~\FIFO{}  is employed, the scheduling plan~$P$ will have each request~$t$ serviced according to the order defined by its arriving time~$a(t)$. In particular, $t_x$ will be processed before~$t_y$ according to~$P$ and in different moments
in time (i.e., they will not be serviced in parallel). 

For the same sequence~$\mathcal{T}$, because $a(t_x) + w(u_x) > a(t_y) + w(u_y)$, \EaS  would invert the order in which tasks~$t_x$ and~$t_y$ are processed, so let us consider the plan~$P'$ that is almost equal to~$P$, having only  the positions of~$t_x$ and~$t_y$ exchanged. Because all the tasks consume the same amount of time, it is clear that we can transform plan~$P$ into  plan~$P^*$ that would be generated by~\EaS  if we apply the same exchange technique sequentially until every pair of requests is positioned accordingly.

Let~$s$ and~$s'$ be the user happiness state vectors
of~$p$ after the execution of plans~$P$ and~$P'$, respectively, and let~$f_x$
and~$f_y$ be the times at which~$t_x$ and~$t_y$ have their processing tasks
finished according to plan~$P$, respectively (\emph{i.e.}, $f_x < f_y$). Let us
refer to~$e(t_x)$ and~$e(t_y)$ as~$e^1(t_x)$ and~$e^1(t_y)$ for  \FIFO, respectively, and as~$e^2(t_x)$ and~$e^2(t_y)$ for  \EaS, respectively.

Finally, let $q_{x,y}: \mathbb{R} \times \mathbb{R} \rightarrow \mathbb{R}$ be
the function parameterized by~$e(t_x)$ and~$e(t_y)$ denoting the sum of the
changes in the happiness levels of users~$u_x$ and~$u_y$ after tasks $t_x$
and~$t_y$ have been serviced, respectively. It is
clear that~$q_{x,y}$ depends on the behaviour of~$i$.

\begin{proposition}
If $q_{x,y}$ is always the same $\forall x, y \in \mathcal{U}$, is monotonic,
and respects exactly one of the following scenarios, then it is possible to
decide if either~\EaS  or~\FIFO  yields a plan
resulting in a user happiness state~$s$ with maximal $||s||_{1}$:
\begin{itemize}
\item $f_{x,y}(a,b) \geq f_{x,y}(c,d)$ whenever $|a| + |b| \geq |c| + |d|$; or 
\item $f_{x,y}(a,b) \leq f_{x,y}(c,d)$ whenever $|a| + |b| \geq |c| + |d|$; or
\item $f_{x,y}(a,b) = f_{x,y}(c,d)$ whenever $|a| + |b| \geq |c| + |d|$.
\end{itemize}
\end{proposition}
\begin{proof}
Simple inspection shows that $a(t_x) + \Delta + w(u_x)$, $a(t_y) + \Delta
+ w(u_y)$, $f_x$, and $f_y$ can appear in six different relative ordering
schemes (e.g., $a_y + \Delta + w(u_y) < a_x + \Delta + w(u_x) < f_x
< f_y$ is one of them)\footnote{Recall that $a(t_x)+\Delta+w(u_x)$ is
already defined as greater than $a(t_y)+\Delta+w(u_y)$.}.  Moreover, one
can also see that $e^1(t_x) + e^1(t_y) = e^2(t_x) + e^2(t_y)$ and that
$max(e^1(t_x), e^1(t_y)) > max(e^2(t_x), e^2(t_y))$ in each of these cases.
Therefore, we have  $|e^1(t_x)| + |e^1(t_y)| \geq |e^2(t_x)| + |e^2(t_y)|$.

Based on these observations and on our hypothesis, we have the following
situations:
\begin{itemize}
\item if $f_{x,y}(a,b) \geq f_{x,y}(c,d)$ whenever $|a| + |b| \geq |c| + |d|$, 
then~$c(s) \geq c(s')$;
\item if $f_{x,y}(a,b) \leq f_{x,y}(c,d)$ whenever $|a| + |b| \geq |c| + |d|$,
then~$c(s) \leq c(s')$; and
\item if $f_{x,y}(a,b) = f_{x,y}(c,d)$ whenever $|a| + |b| \geq |c| + |d|$,
then~$c(s) = c(s')$.
\end{itemize}
Therefore, $P'$ is better than, equal to, or worse than~$P$
if~$f_{x,y}$ has the first, the second, or the third property, respectively.

Finally, if we assume that $f_{x,y}$ is always the same $\forall x, y$
in~$\mathcal{U}$, the resulting user happiness
state associated to~$P^*$ is better than, equal to, or worse than~$P$
if~$f_{x,y}$ has the first, the second, or the third property, respectively. \qed
\end{proof}

\subsubsection{Periodic Requests.}

This section considers the evolution of the user happiness states 
in scenarios where job request arrive periodically in the service provider. 
In this case, the
effects of~$j$ are relevant and we will assume that the maximisation of the
$L^0$-norm of the user happiness state vector is the optimisation goal.

Let us consider the family~$\mathcal{F}$ of scenarios where~$\mathcal{U} = 2m$
and such that every request $t \in \mathcal{T}$ consumes time~$\Delta$.
Given~$0 < \epsilon \leq \Delta$, we partition~$\mathcal{U}$ in two groups of equal size characterized
as follows:
a) users in~$\mathcal{U}_1$ submit requests at time~$3k\Delta$ and at time~$3k\Delta+\epsilon$
for every $k$ in $\mathbb{N}$; b)
 users in~$\mathcal{U}_2$ submit requests at time~$3k\Delta + 4\epsilon/3$ for every~$k$ in $\mathbb{N}$.

In the schedule generated by~\FIFO, all the requests submitted by users in~$\mathcal{U}_1$ 
are served before the tasks submitted by users in~$\mathcal{U}_2$.
More precisely, 
for each task $t$ submitted at time~$3k\Delta$ we have $r(t) = \Delta$,  
for each task $t$ submitted at time~$3k\Delta + \epsilon$ we have $r(t) = 2\Delta-\epsilon$, and  
for each task $t$ submitted at time~$3k\Delta + 4\epsilon/3$ we have $r(t) = 3\Delta-4\epsilon/3$.
We remark that any task submitted at time $t \geq 3k\Delta$ will only
 be serviced  after every task submitted at some time $t < 3k\Delta$ has been already processed.
Finally, let us denote by~$h_k(u)$ and by~$w_k(u)$ the level of happiness and the waiting time of user~$u$
after the~$k$-th step of the scenarios described above.

\begin{proposition}\label{badFIFO}
There is a family of scenarios where a service provider that starts with~$2m$ users 
is able to keep only~$m$ of them active if~\FIFO is employed, while~\EaS  would allow it to keep all the~$2m$ users.
\end{proposition}
\begin{proof}
Let us consider the instances of family~$\mathcal{F}$ for which~$w(u)$ 
is the average of user~$u$'s last~$b$ waiting times $r(t) - \Delta$ 
(completing with $b-k$ values~$\Delta$ whenever $k \leq b$)
for every user~$u$ in~$\mathcal{U}$, $w(u) = 2\Delta$ for $u \in \mathcal{U}_1$, 
$w(u) = \Delta$ for $u \in \mathcal{U}_2$, and  function $i$ is such that
\begin{equation*}
    i_{k+1}(u,e(u,t))=
    \begin{cases}
      0 , & \text{if}\ e(u,t) \leq  \alpha_x \\
       -(h_0(u)-c(u)  )/(b-1), & \text{otherwise}
    \end{cases}
  \end{equation*}
where $\alpha_x = \Delta$ if $u_x \in \mathcal{U}_1$ and
$\alpha_x = 0$ if $u_x \in \mathcal{U}_2$. Finally, let us also assume that $\epsilon < 3\Delta/4$.

For~$u$ in~$\mathcal{U}_1$, we have
 $e(u,t) = -2\Delta$ for each task~$t$ submitted at time~$3k\Delta$, while
$e(u,t) = -\Delta-\epsilon$ for each task~$t$ submitted at time $3k\Delta+\epsilon$. The value of $w(u)$
converges to $(\Delta-\epsilon)/2$ for these users, and as $(2\Delta - \epsilon) - \Delta - (\Delta-\epsilon)/2 = 
(\Delta-\epsilon)/2 < \Delta$, it follows from the definition of~$i$  that~$h(u)$
remains constant for $u \in \mathcal{U}_1$, as $i_{k}(u,e(u,t))$ will clearly be equal to~$0$ for every $k \in [0,b]$.

In the case of users in~$\mathcal{U}_2$,  
as~$w(u)$ is given by the average of the last~$b$ waiting times, if $k < b$, we have
\begin{eqnarray*}
k(2\Delta - 4\epsilon/3) + (b-k)\Delta &<& b(2\Delta - 4\epsilon/3) \\
(b-k)\Delta &<& (b-k)(2\Delta -4\epsilon/3) \\
\Delta &<& 2\Delta - 4\epsilon/3\\
\epsilon &<& 3\Delta/4,
\end{eqnarray*}
\emph{i.e.}, $w(u) < (2\Delta - 2\epsilon)$ until the $b$-th iteration if $\epsilon < 3\Delta/4$.
Moreover, after~$b-1$ iterations, we will have
\begin{eqnarray*}
h_{b-1}(u) &=& h_{0}(u) + (b-1)(c(u)-h_0(u))/(b-1) \\
&=& h_{0}(u) + c(u)-h_0(u) \\
&=& c(u).
\end{eqnarray*}
It follows that if $\epsilon < 3\Delta/4$, then each user~$u$ in~$\mathcal{U}_2$
will stop sending requests to~$p$ after the $(b-1)$-th iteration.

In the schedule generated by~\EaS, 
tasks of users in~$\mathcal{U}_2$ submitted at time~$3k\Delta + 4\epsilon/3$ are serviced just after
the tasks of users in~$\mathcal{U}_1$
submitted at time~$3k\Delta$.

More precisely, 
for each task $t$ submitted at time~$3k\Delta$ we have $r(t) = \Delta$,  
for each task $t$ submitted at time~$3k\Delta + \epsilon$ we have $r(t) = 3\Delta-\epsilon$, and  
for each task $t$ submitted at time~$3k\Delta + 4\epsilon/3$ we have $r(t) = 2\Delta-4\epsilon/3$.

Under these circumstances, the value of~$h(u)$ for $u \in \mathcal{U}_2$ does not change, as~$w(u)$ decreases monotonically  
from~$\Delta$ to~$\Delta - 4\epsilon/3$ as~$e(u,t)$ converges from $-4\epsilon/3$ to~$0$
over the course of the first~$b$ iterations.

For each user~$u$ in~$\mathcal{U}_1$, 
$h(u)$ do not change, 
as~$w(u)$ decreases monotonically  
from~$2\Delta$ to~$\Delta - \epsilon/2$ and $(3\Delta - \epsilon) - \Delta -  \Delta + \epsilon/2 = \Delta-\epsilon/2 < \Delta$.

Finally, after the first~$b$ iterations, $w(u) = \Delta - 4\epsilon/3$ for $u \in \mathcal{U}_2$ and
$w(u) = \Delta - \epsilon/2$ for $u \in \mathcal{U}_1$. Tasks from~$\mathcal{U}_1$ submitted at time
$3k\Delta$ will continue as the first ones to be serviced. Users from~$\mathcal{U}_1$ submitting tasks
at time $3k\Delta + \epsilon$ will be unsatisfied if they do not get the response by time $3k\Delta + \epsilon + \Delta - \epsilon/2 + \Delta = 3k\Delta + 2\Delta + \epsilon/2$, while
users from~$\mathcal{U}_1$ submitting tasks
at time $3k\Delta + 4\epsilon/3$ will be unsatisfied if they do not get the response by time $3k\Delta + 4\epsilon/3 + \Delta - 4\epsilon/3 + \Delta = 3k\Delta + 2\Delta$. Therefore, \EaS will keep  using
the same ordering that it used in the first iteration once the system reaches its equilibrium

It follows  that there are scenarios for which a service provider employing~\FIFO may only be able to keep~$m$ users where~\EaS would allow it
to reach the equilibrium having all the~$2m$ users active. \qed
\end{proof}

We remark that one can modify~$\mathcal{F}$ and the proof 
above in order to show that a service provider having initially~$qm$ users
$\forall q \in \mathbb{N}$ and computing requests of different processing times
may eventually finish with only~$m$ active users after a finite number of iterations.

\begin{proposition}\label{badPbH}
There is a family of scenarios for which  a service provider that starts with~$2m$ users 
is able to keep only~$m$ of them active if~\EaS is employed, while~\FIFO  would allow it to keep all the~$2m$ users.
\end{proposition}

\begin{proof}
Let us consider instances of family~$\mathcal{F}$ that are similar to the ones
considered in the proof of Proposition~\ref{badFIFO} with the exception of
function~$i$, which is assumed here to be 
\begin{equation*}
    i_{k+1}(u,e(u,t))=
  \begin{cases}
      0, & \text{if}\ e(u,t) \leq  0 \\
       -(h_0(u)-c(u)  )/\beta, & \text{otherwise,}
    \end{cases}
  \end{equation*}
where $\beta = b-1$ for $u \in \mathcal{U}_1$ and
$\beta = b+1$ for $u \in \mathcal{U}_2$.

We also assume that $w(u) = 3\Delta/2$ for users in~$\mathcal{U}_1$, that 
$w(u) = \Delta$ for users in $\mathcal{U}_2$, that $w(u)$ is still given by the
average of last~$b$ waiting times $r(t)-\Delta$ (completing with $b-k$ values $3\Delta/2$ and~$\Delta$
whenever $k < b$ for~$u$ in~$\mathcal{U}_1$ and~$\mathcal{U}_2$, respectively), 
that users in $\mathcal{U}_2$ submit requests at time $3k\Delta + 2\epsilon$ 
instead of requesting at time $3k\Delta + 4\epsilon/3$
for every $k$ in $\mathbb{N}$,
and  that $\epsilon < \Delta/2$

Let us consider~\FIFO first.
Using arguments similar to the ones employed in Proposition~\ref{badFIFO}, one
can show that users in~$\mathcal{U}_1$ will have their happiness levels unchanged
and that~$w(u)$ will be $(\Delta-\epsilon)/2$ after the $i$-th iteration.
 
For users in~$\mathcal{U}_2$,
we can see below that~$w(u)$ will not become 
$2\Delta - 2\epsilon$ before the $b$-th iteration if we assume that $\epsilon < 2\Delta$:
\begin{eqnarray*}
k(2\Delta - 2\epsilon) + (b-k)\Delta &<& b(2\Delta - 2\epsilon) \\
(b-k)\Delta &<& (b-k)(2\Delta - 2\epsilon) \\
\Delta &<& 2\Delta - 2\epsilon\\
\epsilon &<& \Delta/2.
\end{eqnarray*}
Moreover, we also have that 
\begin{eqnarray*}
h_{b-1}(u) &=& h_{0}(u) + (b-1)(c(u)-h_0(u))/(b+1) \\
&=& 2h_{0}(u)/(b+1) +   (b-1)c(u)/(b+1).
\end{eqnarray*}
Because $h_0(u)  >c(u)$, it follows that $h_{b-1}(u)  > c(u)$. In the following
iterations, the system reaches an equilibrium (i.e., $e(u,t) = 0$ for
any task~$t$ submitted by user~$t$ if $a(t) \geq 3b\Delta)$, so we have that users in group~$\mathcal{U}_2$ will 
still active.

In the case of~\EaS, tasks submitted by users in ~$\mathcal{U}_1$ at time~$3k\Delta + \epsilon$
will have waiting time $2\Delta - \epsilon$, a value that will be superior to~$w(u)$ in the first~$b-1$ iterations
because $\epsilon < \Delta/2$. Therefore, in the $b$-th step, all the users from~$\mathcal{U}_1$ will have the
critical level of their happiness levels surpassed, and therefore they will abandon the service provider.

It follows  that there are scenarios for which a service provider employing~\EaS may only be able to keep~$m$ users where~\FIFO would allow it
to reach the equilibrium having all the~$2m$ users active. \qed

\end{proof}
 
Finally, the last result shows that it is possible (and easy) to identify families of scenarios where several groups of
users submit requests and eventually only~$m$ users will remain if we employ 
either~\FIFO or~\EaS.

\begin{proposition} %\label{badPbH}
There are scenarios for which both~\EaS and~\FIFO  will lead
to a situation where only~$m$ costumers will keep using the service.
\end{proposition}
\begin{proof}
Let us consider the family~$\mathcal{F}$ of scenarios where~$2m$ users submit requests to a service provider
and  every request $t \in \mathcal{T}$ consumes time~$\Delta$.
Given~$0 < \epsilon \leq \Delta$, we partition~$\mathcal{U}$ in two groups of equal size:
a) users in~$\mathcal{U}_1$ submit requests at time~$2k\Delta$ for every $k$ in $\mathbb{N}$; and
b) users in~$\mathcal{U}_2$ submit requests at time~$2k\Delta +\epsilon$ for every~$k$ in $\mathbb{N}$.

Let us consider instances of family~$\mathcal{F}$, and let us assume that~$w(u)$ 
is given by the average of user~$u$'s last~$b$ waiting times $t(t) - \Delta$ 
(completing with $b-k$ values~$\Delta$ whenever $k \leq b$)
for every user~$u$ in~$\mathcal{U}$. Let us assume that $w(u) = 2\Delta$ for $u \in \mathcal{U}_1$, 
that $w(u) = 0$ for $u \in \mathcal{U}_2$, and that function $i$ is such that
\begin{equation*}
    i_{k+1}(u,e(u,t))=
    \begin{cases}
      0 , & \text{if}\ e(u,t) \leq  \alpha_x \\
       -(h_0(u)-c(u)  )/(b-1), & \text{otherwise}
    \end{cases}
  \end{equation*}
where $\alpha_x = \Delta$ if $u_x \in \mathcal{U}_1$ and
$\alpha_x = 0$ if $u_x \in \mathcal{U}_2$.

It is clear that tasks from users in~$\mathcal{U}_1$ will always be served as soon as they arrive
while tasks from users in~$\mathcal{U}_2$ will have $2\Delta - \epsilon$ as response time
if we use either~\FIFO or~\EaS, 
As we already showed in the proof of Theorem~\ref{badFIFO}, the happiness level of users
in~$\mathcal{U}_2$ in this case will get below the critical point before~$w(u)$ becomes $2\Delta - \epsilon$,
so these users will abandon the system.

Therefore, we conclude that, in the worst case, both~\FIFO and~\EaS  will
leave a service provider with only~$m$ active users.  \qed
\end{proof}
The families of instances described in the proof above represent the typical worst-case scenario
for online job scheduling, where decisions that are taken at a certain point in time may lead to
bad situations that cannot be modified in non-preemptive systems.

% -------------------------------------------------------------------------------------------

\section{Evaluation}
\label{sec:evaluation}

\newcommand{\maxNumberUsers}{100\xspace}
\newcommand{\thinkTime}{100 seconds\xspace}
\newcommand{\jobSize}{10 seconds\xspace}
\newcommand{\ewmaWindowSize}{20\xspace}
\newcommand{\ewmaAlpha}{0.8}
\newcommand{\windowToDeleteOutliers}{4\xspace}
\newcommand{\expectationMax}{60 seconds\xspace}
\newcommand{\expectationMin}{40 seconds\xspace}

In addition to the analytical investigation provided beforehand, this section presents simulation results that evaluate the performance of the scheduling strategies under different workload conditions.

\subsection{Environment Setup}

A discrete event simulator built in house was used to evaluate the performance of the scheduling strategies. To model the load of a Cloud service, we crafted three types of workloads with variable numbers of users over a 24-hour period as shown in Figure~\ref{fig:workloads}. The rationale behind the workloads is described as follows:

\begin{figure*}[!ht]
\centerline{\includegraphics[width=0.9\textwidth]{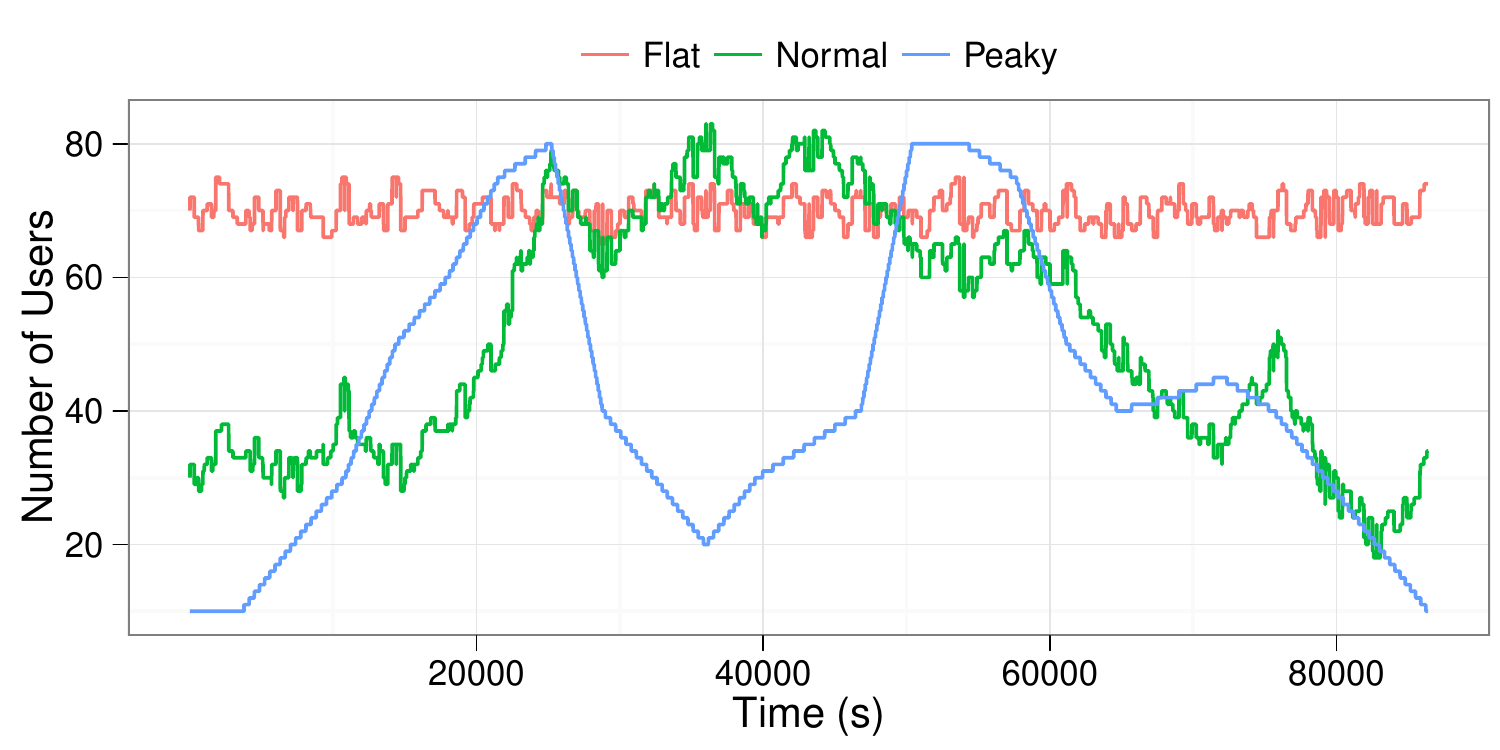}}
\vspace{5mm}
\caption{The three types of workloads.}
\label{fig:workloads}
\end{figure*}

\begin{itemize}

\item \emph{Normal day}: consists of small peaks of utilisation during the start, middle, and end of work hours reflecting the time when users check their e-mails and websites, for example. Outside these intervals, but still in work hours, this workload remains around the peak values, while outside the working hours it goes down significantly.

\item \emph{Flat day}: consists of a flat number of users during the whole period. Although unrealistic in most real-world environments, this load is used to evaluate a scenario with near constant load.

\item \emph{Peaky day}: consists of tipping workload peaks, a configuration that is realistic and reflects the situation where impacting news reach the outside world, causing users to access a service more often. The configuration is used to test the solution's behaviour handling stress situations. 

\end{itemize}

For each workload we vary the number of resources used by the Cloud service, thus allowing for evaluating the system under different stress levels. When using the system, a user makes a request and waits for its results before making a new request, with a think time between receiving results and making another request uniformly distributed between 0 and \thinkTime. To facilitate the analysis and comparison among the techniques, the length of jobs is constant (\jobSize). 

Previous interactions with the service are used to build a user's expectation on how the service should respond, and how quickly a request should be processed. The model that defines a user's expectation on the response time of a request uses two moving averages, (i) an Exponential Weighted Moving Average (EWMA) of the previous \ewmaWindowSize response times, with $\alpha = \ewmaAlpha$; and (ii) an average of the past \windowToDeleteOutliers response times, used to eliminate outliers. When a request completes, if the response time is 30\% below the average of the past \windowToDeleteOutliers response times, then the EWMA is not updated, though the value is considered in future iterations. In essence, this model states that the user expects the service to behave similar to previous interactions, with a higher weigh to more recent requests. Even though changes in response time affect the user's perception of the service, she disregards large deviations in service quality; unless they become common. As we believe that in real conditions, users would not correctly average their past response times (\textit{i.e.} they may not recall past experiences well) we add tolerance of 20\% to the estimate of response time provided by the model.

We consider that users have different levels of patience --- If you ever listened to customers' complaints in a supermarket queue in a busy city like Sao Paulo you probably know what we are referring to. Hence, a user's response time threshold---\textit{i.e.} the maximum response time that she considers acceptable---is randomly selected between \expectationMin and \expectationMax. The provider stores information on how it served previous requests made by a user and users the same model described above to compute an estimate of what it believes the user's expectation to be. \expectationMax is also what the provider considers to be the maximum acceptable response time that satisfies the service users. However, for \EaS and \PaS, if a request's response is above \expectationMax, the EWMA is updated with \expectationMin, which may give the user priority the next time she submits a request. It is a way the scheduler finds to penalise itself for yielding a response time too far from what it believes the user's expectation to be.

\subsection{Result Analysis}

\setlength{\subfigtopskip}{0pt}
\setlength{\subfigcapskip}{0pt}
\setlength{\subfigbottomskip}{1pt}

Figure \ref{fig:unhappy} depicts the Patience Indexes (as defined in Section \ref{sec:patienceaware}) of requests when below 1.0 for flat, normal, and peaky workloads. The lower the values the more unhappy the users. We observe that for high and low system load ({\it i.e.} r4--6 and r16--20), all strategies perform similarly, whereas for the other loads \PaS and \EaS produce higher Patience Indexes than \FIFO. Under high loads, most requests are completed after the expected response time, thus not allowing the scheduler to exchange the order of the requests in the waiting queue in subsequent task submissions. On the other hand, a very light system contains a short (or empty) waiting queue; hence not having requests to be sorted.

The impact of the scheduling strategies becomes evident when the system is almost fully loaded, \textit{i.e.} when the waiting queue is not empty and there are requests that can quickly be assigned to resources. In this scenario, requests with longer response time expectations can give room to tasks from impatient users. The \FIFO strategy does not explore the possibility of modifying the order of requests considering user patience.

\begin{figure}[!th]

        \centering
        \subfigure[Flat.]{
        \includegraphics[width=0.47\textwidth]{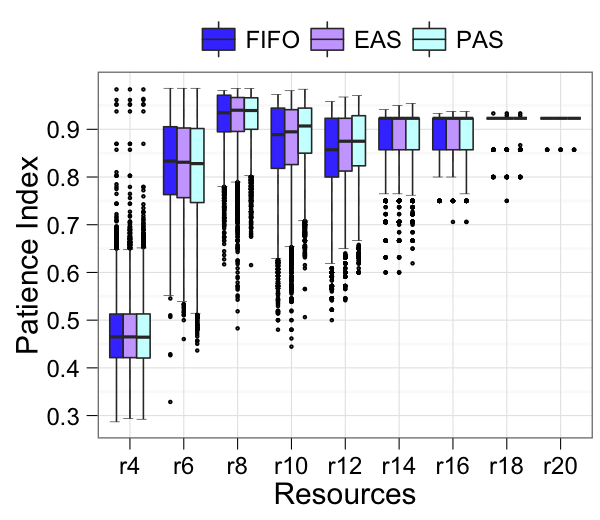}
                }
        \subfigure[Normal.]{
        \includegraphics[width=0.47\textwidth]{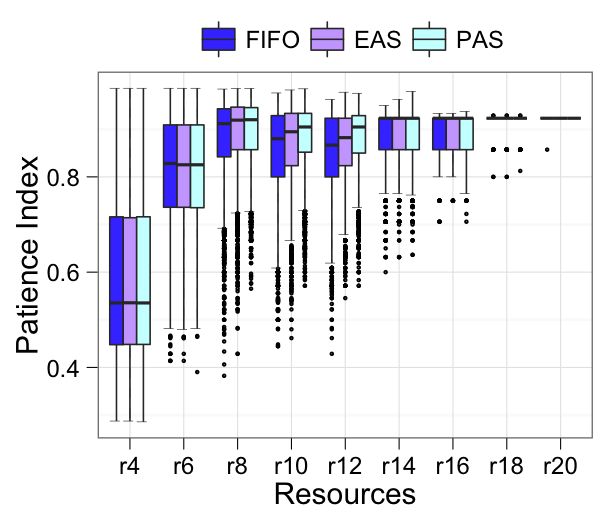}
                }
        \subfigure[Peaky.]{
        \includegraphics[width=0.47\textwidth]{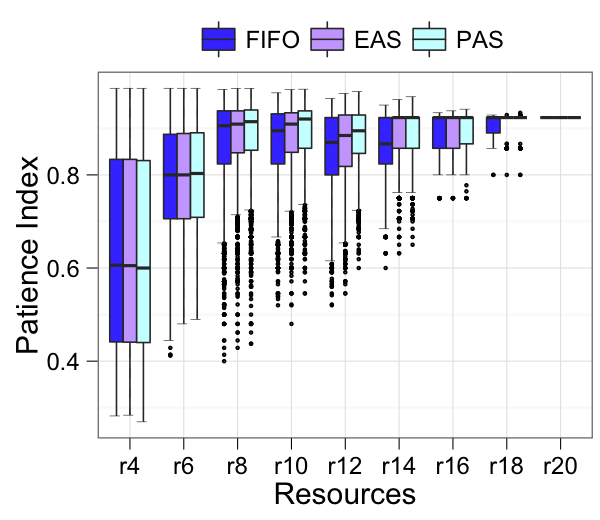}
        }
        \caption{Patience index under different workloads.}
        \label{fig:unhappy}
\end{figure}

Figure \ref{fig:number} presents the percentage of requests that were served considerably later than the expected response time, that is, when their Patience Index tends to zero. Such requests represent the stage where users' level of happiness is decreasing considerably. The percentage was normalised by the total number of requests for each resource setting for all strategies. The behaviour of this metric follows the patience indexes, but it highlights the impact of the proposed strategies have on users with very low patience levels.

\begin{figure}[!th]

        \centering
        \subfigure[Flat.]{
        \includegraphics[width=0.47\textwidth]{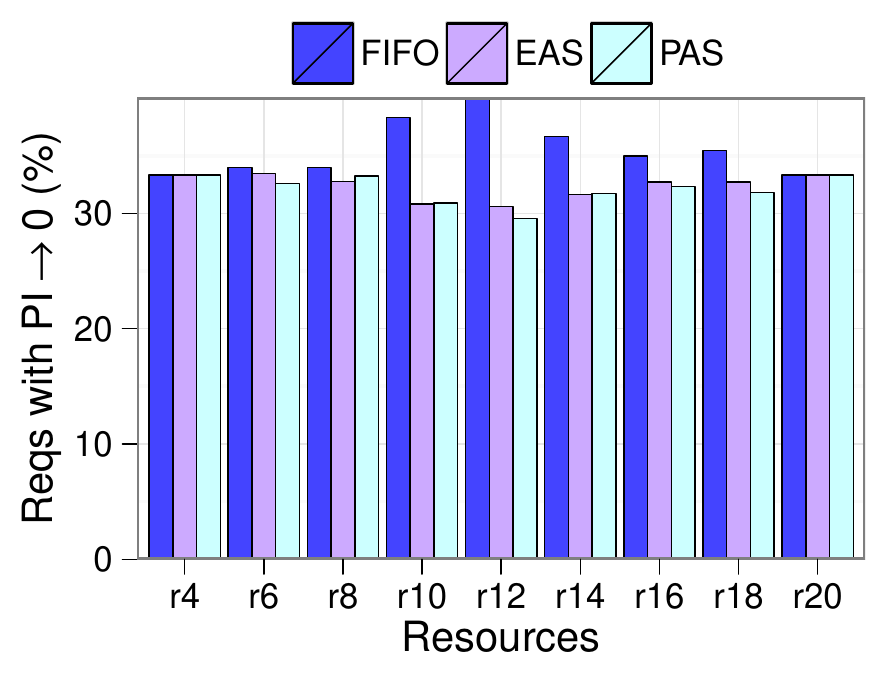}
                }
        \subfigure[Normal.]{
        \includegraphics[width=0.47\textwidth]{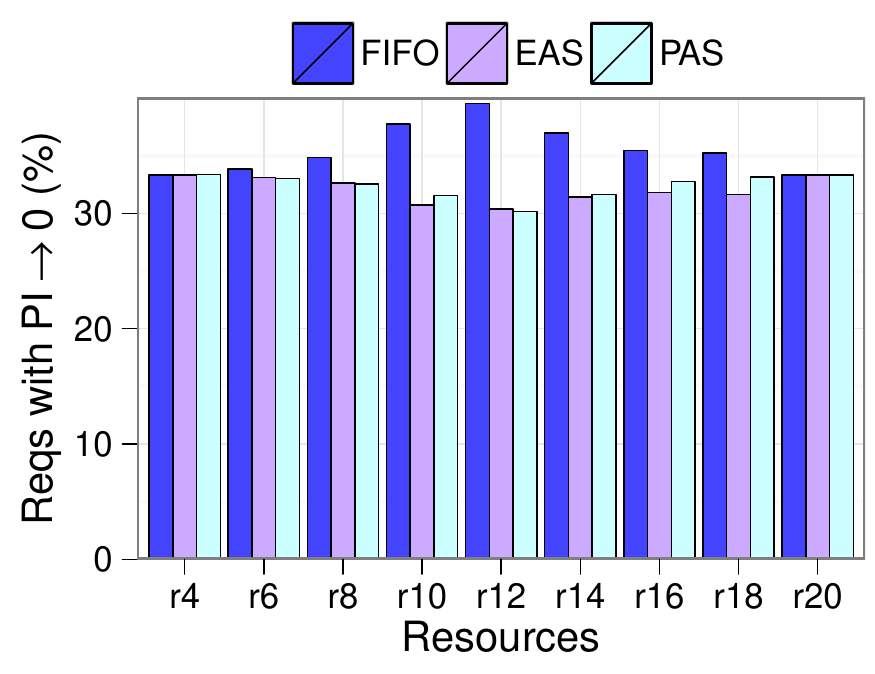}
                }
        \subfigure[Peaky.]{
        \includegraphics[width=0.47\textwidth]{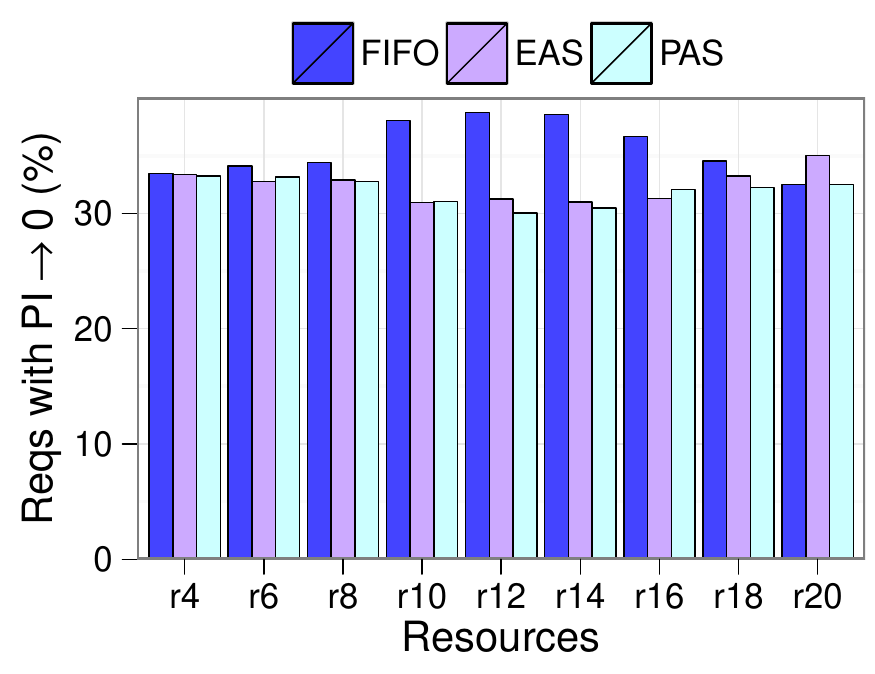}
        }
        \caption{\% of requests whose patience index tends to 0.}
        \label{fig:number}
\end{figure}

% -------------------------------------------------------------------------------------------

\section{Related Work}
\label{sec:related_work}

Scheduling is a well-studied topic in several domains, including resource management for clusters, grids, and more recently Cloud computing. Commonly used algorithms include First-In First-Out, priority-based, deadline-driven, some hybrids using backfilling techniques \cite{tsafrir2007backfilling}, among others \cite{braun2001comparison,feitelson1997theory}. In addition to priority and deadline, other factors have been considered, such as fairness \cite{doulamis2007fair}, energy-consumption \cite{pineau2011energy}, and context-awareness \cite{assuncao2012contextaware}. Moreover, utility functions were used to model how the importance of results to users varies over time \cite{lee2007precise,auyoung2006service} and attention scarcity was leveraged to determine priority of service requests in the Cloud \cite{netto2013leveraging}.

User behaviour has been explored for optimising resource management in the context of Web caching and page pre-fetching\cite{galletta2004website,alt2012bridging,atterer2006knowing,cunha1997determining}. The goal is to understand how users access web pages, investigate their tolerance level on delays, and pre-fetch or modify page content to enhance user experience. Techniques in this area focus mostly on web content and minimising response time of user~requests.

Service research has also investigated the impact of delays on users' behaviour. For instance, Taylor \cite{taylor1994waiting} described the concept of delays and surveyed passengers affected by delayed flights to test their hypotheses. Brown et al. \cite{brown2005statistical} and Gans et al. \cite{gans2003telephone} investigated the impact of service delays in call centres. In behavioural economics, Kahneman and Tversky \cite{kahneman1979prospect} introduced prospect theory to model how people make choices in situations that involve risk or uncertainty.

% -------------------------------------------------------------------------------------------

\section{Conclusions}
\label{sec:conclusion}

We presented \acf{PAS} and \acf{EAS} strategies that use estimates on users' level of tolerance or patience to define the order in which resources are assigned to requests.

We compared the \EaS with \FIFO analytically and showed that it is not trivial to choose between both algorithms. In fact, the quality of the scheduling plans they produce depends strongly on users' level of happiness with a service and tolerance to delays. Deeper analytical results will probably require a better understanding and more precise characterisation of these two aspects. Our computational evaluation shows that both \PaS and \EaS perform better than \FIFO under peak load scenarios, and that \PaS is slightly better than \EaS. 

Several aspects can be explored in future work. The~\PaS strategy works basically as a greedy algorithm, and in spite of the challenges involving the prediction of resolution times for tasks that are still in the queue, we believe that the use of more advanced data structures and/or algorithms may improve the quality of its scheduling plans.

%We also proved that there are scenarios for which all the proposed strategies are not able to deal satisfactorily with worst-case situations for online job scheduling. In a future work, we intend to develop a stochastic scheduling strategy that considers historical data to predict the arrival time of requests from each user. Partial knowledge about future events that can be obtained with the support of classical tools from statistics and machine learning may provide information that will lead to more robust scheduling strategies.

\bibliographystyle{splncs03}
\bibliography{references}
\end{document}